\documentclass[10pt,a4paper]{amsart}
\usepackage{amsaddr}
\usepackage[utf8]{inputenc}   
\usepackage[english]{babel}   
\usepackage{amsmath, amsthm, amssymb, amsfonts,bbm}
\usepackage{graphicx}
\usepackage{wrapfig}
\usepackage[top=2cm,bottom=3cm,left=3cm,right=3cm]{geometry}
\usepackage{subfig}
\usepackage{url}
\usepackage[
    pdftex,
    pdfstartview={XYZ 0 1000 1.0},
    colorlinks,
    breaklinks=true,
    citecolor=blue,
    filecolor=blue,
    bookmarks=false,
    linkcolor=black,
    urlcolor=blue,
    pdfborder={0 0 0}
]{hyperref}
\usepackage{multirow}
\usepackage{diagbox}
\usepackage{latexsym}  % Symboles
\usepackage{epstopdf}
\usepackage{color} 
\usepackage{fancyhdr}
\usepackage{nameref}
\usepackage{xspace} 
\usepackage{ulem}

\newtheorem{theorem}{Theorem}[section]
\newtheorem{lemma}[theorem]{Lemma}
\newtheorem{corollary}[theorem]{Corollary}

\newtheorem{definition}[theorem]{Definition}
\newtheorem{proposition}[theorem]{Proposition}

\newcommand{\Dm}{{\mathbb D}}
\newcommand{\Hm}{{\mathbb H}}
\newcommand{\Rm}{{\mathbb R}}
\newcommand{\R}{{\mathbb R}} 
\newcommand{\Sm}{{\mathbb S}}
\newcommand{\Tm}{{\mathbb T}}

\newcommand{\Z}{{\mathbb Z}}

\newcommand{\gp}{G} % groupe d'isométries
\newcommand{\e}{\mathbbm{1}_{\gp}\xspace} %identité
\newcommand{\is}{g} % element du groupe d'isométries

\newcommand{\sur}{{\mathcal{M}_2}} % surface générique (tore plat ou hyperbolique)

\newcommand{\lift}[1]{\widetilde{#1}} % un représentant quelconque dans l'orbite
\newcommand{\rep}[1]{{{\widetilde{#1}}^0}} % le représentant unique dans le domaine fondamental de départ
\newcommand{\F}{{\Omega^0}} % le domaine fondamental donné au départ

\newcommand{\proj}{{\rho}}
\newcommand{\pb}[1]{{\proj^{-1}(#1)}} % pull-back

\newcommand{\diam}{{\Delta}} % diamètre

\newcommand{\cF}{{\mathcal F}}

\newcommand{\cgal}{{\sc cgal}}

\definecolor{dgreen}{RGB}{0,150,0}

\title{Flipping Geometric Triangulations on Hyperbolic Surfaces}\thanks{The authors were supported by the grant(s)
    ANR-17-CE40-0033 of the French National Research Agency ANR (project SoS)
    and
    INTER/ANR/16/11554412/SoS of the Luxembourg National Research fund
    FNR ({\small\url{https://members.loria.fr/Monique.Teillaud/collab/SoS/}}).}
\date{\today}
\author{Vincent Despré}
\address{Université de Lorraine, CNRS, Inria, LORIA, F-54000 Nancy, France}
\email{vincent.despre@loria.fr}

\author{Jean-Marc Schlenker}
\address{Department of Mathematics, University of Luxembourg, Luxembourg}
\email{jean-marc.schlenker@uni.lu}

\author{Monique Teillaud}
\address{Université de Lorraine, CNRS, Inria, LORIA, F-54000 Nancy, France} 
\email{monique.teillaud@inria.fr}

\begin{document}
\maketitle

\begin{abstract}
We consider geometric triangulations of surfaces, i.e., triangulations
whose edges can be realized by disjoint locally geodesic segments. We prove
that the flip graph of geometric triangulations with fixed vertices of
a flat torus or a closed hyperbolic surface is connected. We give
upper bounds on the number of edge flips that are necessary to
transform any geometric triangulation on such a surface into a
Delaunay triangulation.
\end{abstract}

\section{Introduction}
In this paper, we investigate triangulations of two categories of surfaces: flat tori, i.e., surfaces of genus 1 with a locally Euclidean metric, and hyperbolic surfaces, i.e., surfaces of genus at least 2 with a locally hyperbolic metric (these surfaces will be introduced more formally in Section~\ref{sec:surface}). 

Triangulations of surfaces can be considered in a purely topological manner: a triangulation of a surface is a graph whose vertices, edges and faces partition the surface and whose faces have three (non-necessarily distinct) vertices. However, when the surface is equipped with a
Euclidean or hyperbolic structure, it is possible to consider {\em geometric} triangulations, i.e., triangulations whose edges can be realized as interior disjoint locally geodesic segments (Definition~\ref{def:geometric}). Note that a geometric triangulation can still have loops and multiple edges, but no contractible loop and no contractible cycle formed of two edges.
We will prove that any Delaunay triangulation (Definition~\ref{def:Delaunay}) of the considered surfaces is geometric (Proposition~\ref{pr:geometric}). 

The flip graph of triangulations of the Euclidean plane has been well studied. It is known to be connected; moreover the number of edge flips that are needed to transform any given triangulation with $n$ vertices in the plane into the Delaunay triangulation has complexity $\Theta(n^2)$~\cite{ferran}. We are interested in generalizations on this result to surfaces. Flips in triangulations of surfaces will be defined precisely later (Definition~\ref{def:flip}), for now we can just think of them as similar to edge flips in triangulations of the Euclidean plane. We emphasize that geodesics only locally minimize the length, so the edges of a geometric triangulation are generally not shortest paths. We will prove that the number of geometric triangulations on a set of points can be infinite, whereas the flip graph of "shortest path" triangulations is small but not connected in most situations~\cite{cghmsv-2010}. 

\begin{definition}\label{def:flip-graph}
Let $(\sur,h)$ be either a torus $(\Tm^2,h)$ equipped with a Euclidean structure $h$ or a closed oriented surface $(S,h)$ equipped with a hyperbolic structure $h$. Let $V\subset\sur$ be a set of $n$ points. 
The {\em geometric flip graph} $\cF_{\sur,h,V}$ of $(\sur,h,V)$ is the graph whose vertices are the geometric triangulations of $(\sur,h)$ with vertex set $V$ and where two vertices are connected by an edge if and only if the corresponding triangulations are related by a flip. 
\end{definition}

Our results are mainly interesting in the hyperbolic setting, which is richer than the flat setting. However, to help the readers' intuition, we also present them for flat tori, where they are slightly simpler to prove and might even be considered as folklore. The geometric flip graph is known to be connected for the special case of flat surfaces with conical singularities and triangulations whose vertices are these singularities~\cite{t-gtf-19}.

The main results of this paper are:
\begin{itemize}
\item The geometric flip graph of $(\sur,h,V)$ is connected (Theorems~\ref{th:flip-torus} and~\ref{th:flip-hyp}).
\item The Delaunay triangulation can be reached from any geometric triangulation  by a path in the geometric flip graph $\cF_{\sur,h,V}$ whose length is bounded by $n^2$ times a quantity measuring the \emph{quality} of the input triangulation (Theorems~\ref{tm:bound-torus} and~\ref{tm:bound-hyp}). 
\end{itemize}
If an initial triangulation of the surface only having one vertex is given, then the Delaunay triangulation can thus be computed incrementally by inserting points one by one in a very standard way: for each new point, the triangle containing it is split into three, then the Delaunay property is restored by propagating flips. This approach, based on flips, can handle triangulations of a surface with loops and multiarcs, which is not the case for the approach based on Bowyer's incremental algorithm~\cite{ct-dtced-16,btv-dtosl-16}. 
The work presented here can hardly be compared with broad results on computing Delaunay triangulations on very general manifolds~\cite{jdandco}.

\section{Background and notation}

\subsection{Surfaces}\label{sec:surface}
In this section, we first recall a few notions, then we illustrate them for the two classes of surfaces (flat tori and hyperbolic surfaces) that we are interested in. 

Let $\sur$ be a closed oriented surface, i.e., a compact connected oriented 2-manifold without boundary. 
There is a unique simply connected surface $\lift{\sur}$, called the {\em universal cover} of $\sur$, equipped
with a projection $\proj:\lift{\sur}\to \sur$ that is a local  diffeomorphism.
There is a natural action on $\lift{\sur}$ of the fundamental group $\pi_1(\sur)$ of $\sur$ so that for all $p\in \sur$, $\proj^{-1}(p)$ is an orbit under the action of $\pi_1(\sur)$. We will denote as $\lift{p}$ a {\em lift} of $p$, i.e., one of the elements of the orbit $\proj^{-1}(p)$.
A {\em fundamental domain} in $\lift{\sur}$ for the action of $\pi_1(\sur)$ on $\lift{\sur}$ is a connected subset $\Omega$ of $\lift{\sur}$ that intersects each orbit in exactly one point, or, equivalently, such that the restriction of $\proj$ to $\Omega$ is a bijection from $\Omega$ to $\sur$~\cite{massey}. 
The genus $g$ of $\sur$ is its number of handles. In this paper, we consider surfaces with constant curvature ($0$ or $-1$). The value of the curvature is given by Gauss-Bonnet Theorem and thus only depends on the genus: a surface of genus $0$ only admits spherical structures (not considered here); a flat torus is a surface of genus $1$ and admits Euclidean structures; a surface of genus $2$ and above admits only hyperbolic structures (see below).

From now on, $\sur$ will denote either a flat torus or a closed hyperbolic surface. 

\paragraph*{Flat tori.}
We denote by $\Tm^2$ the topological torus, that is, the product $\Tm^2=\Sm^1\times \Sm^1$ of two copies of the circle. Flat tori are obtained by taking the quotient of the Euclidean plane by an Abelian group generated by two independent translations. There are in fact many different Euclidean structures on $\Tm^2$; if one considers Euclidean structures up to homothety -- which is sufficient for our purposes here -- a Euclidean structure is uniquely determined by a vector $u$ in the upper half-plane $\R\times \R_{>0}$: to such a vector $u$ is associated the Euclidean structure
$ (\Tm^2, h_u)\sim \R^2/(\Z e_1+\Z u)~,$ where $e_1=(1,0)$ and $u=(u_x,u_y)\in \R^2$ is linearly independent from $e_1$.  The orbit of a point of the plane is a lattice. The area $A_h$ of the surface is $|u_y|$. 
The plane $\R^2$, equipped with the Euclidean metric, is then isometric to the universal cover of the corresponding quotient surface.

\paragraph*{Hyperbolic surfaces.}
We now consider a closed oriented surface $S$ (a compact oriented surface without boundary) of genus $g\geq 2$. Such a surface does not admit any Euclidean structure, but it admits many {\em hyperbolic} structures, corresponding to metrics of constant curvature $-1$, locally modeled on the hyperbolic plane $\Hm^2$. 
Given a hyperbolic structure $h$ on $S$, the surface $(S,h)$ is isometric to the quotient $\Hm^2/\gp$, where $\gp$ is a (non-Abelian) discrete subgroup of the isometry group $PSL(2,\R)$ of $\Hm^2$ isomorphic to the fundamental group $\pi_1(S)$. The universal cover $\lift{S}$
is isometric to the hyperbolic plane $\Hm^2$. 

For completeness, we recall below some properties of the hyperbolic plane.

\subsection{The Poincaré disk model of the hyperbolic plane} 
\label{sec:hyperb}

In the Poincaré disk model~\cite{berger}, the hyperbolic plane is represented as the open unit disk $\Dm^2$ of $\R^2$. The points on the unit circle represent points at infinity. The geodesic lines consist of circular arcs contained in the disk $\Dm^2$ and that are orthogonal to its boundary (Figure~\ref{fig:poincare} (left)). The model is conformal, i.e., the Euclidean angles measured in the plane are equal to the hyperbolic angles. 

We won't need the exact expression of the hyperbolic metric here.
However, the notion of hyperbolic circle is relevant to us. Three non-collinear points in the hyperbolic plane $\Hm^2$ determine a {\em circle}, which is the restriction to the Poincaré disk of a Euclidean circle or line. If $C$ is a Euclidean circle or line and $\phi:\Dm^2\to \Dm^2$ is an isometry of the hyperbolic plane, then $\phi(C\cap \Dm^2)$ is still the intersection with $\Dm^2$ of a Euclidean circle or a line.

A key difference with the Euclidean case is that the ``circle'' defined by 3 non-collinear points in $\Hm^2$ is generally not compact (i.e., it is not included in the Poincaré disk). The compact circles are sets of points at constant (hyperbolic) distance from a point. Non-compact circles are either horocycles or hypercycles, i.e., connected components of the set of points at constant (hyperbolic) distance from a hyperbolic line (Figure~\ref{fig:poincare} (right))~\cite{gardner}.\footnote{A synthetic presentation can be found at \url{http://en.wikipedia.org/wiki/Hypercycle_(geometry)}}
Therefore, the relatively elementary tools that can be used for flat tori must be refined for hyperbolic surfaces. Still, some basic properties of circles still hold for non-compact circles. A non-compact circle splits the hyperbolic plane into two connected regions. We will call \emph{disk}  the region of the corresponding Euclidean disk %(limited by the associated Euclidean circle) 
that lies in the Poincaré disk. When a non-compact circle is determined by the three vertices of a triangle, its associated disk is convex (in the hyperbolic sense) and contains the whole triangle. 

\begin{figure}[htb]
    \centering
    \includegraphics{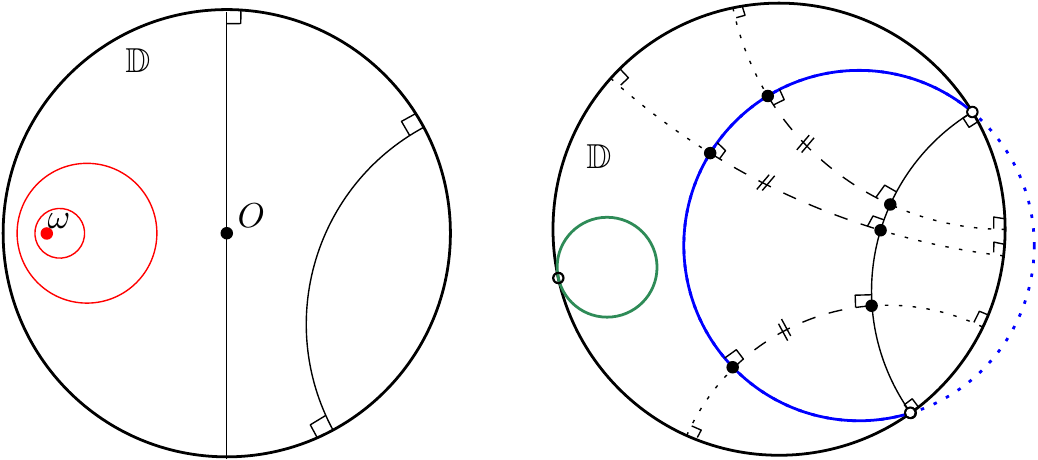}
    \caption{The Poincaré disk. Left: Geodesic lines (black) and compact circles (red) centered at point $\omega$. Right: A horocycle (green). A hypercycle (blue), whose points have constant distance from the black geodesic line.}
    \label{fig:poincare} 
\end{figure} 

Triangulations of hyperbolic spaces have been studied~\cite{bdt-hdcvd-14} and implemented in \cgal\ in 2D~\cite{cgal-bit-ht2}. Note that that previous work was not considering non-compact circles as circles.

\subsection{Triangulations on surfaces}\label{sec:delaunay}
Let $(\sur,h)$ be either a torus $(\Tm^2,h)$ equipped with a Euclidean structure $h$ or a closed surface $(S,h)$ equipped with a hyperbolic structure $h$. Let $V\subset\sur$ be a finite subset of points, and let $T$ be a triangulation of $\sur$ with vertex set $V$.

Recall that given two distinct points $v, w\in \sur$, any homotopy class of paths on $\sur$ with endpoints $v$ and $w$ contains a unique locally geodesic segment. We can recall the following simple notion of geometric triangulation.

\begin{definition}\label{def:geometric} A triangulation $T$ on $\sur$ is said to be {\em geometric} for $h$ if it can be realized with interior disjoint locally geodesic segments as edges.  
\end{definition}

If $T$ is a triangulation of $\sur$, its inverse image\footnote{the notion of \emph{pull-back} would be more correct but we stay with inverse image for simplicity} $\pb{T}$ is the (infinite) triangulation of $\lift{\sur}$ with vertices, edges and faces that are connected components of the lifted images by $\proj^{-1}$ of the vertices, edges and faces of $T$.
\begin{definition}\label{def:diam}
The {\em diameter} $\diam(T)$ of $T$ is the smallest diameter of a fundamental domain that is the union of lifts of the triangles of $T$ (with geodesic edges) in $\lift{\sur}$.
\end{definition}
The diameter $\diam(T)$ is not smaller than the diameter of $(S,h)$. It is unclear how to compute $\diam(T)$ algorithmically and the problem looks difficult. However bounds are easy to obtain: $\diam(T)$ is at least equal to the maximum of the diameters of the triangles of $\pb{T}$ in $\lift{\sur}$ and is at most the sum of the diameters of these triangles.

\begin{definition} \label{def:Delaunay} We say that a triangulation $T$ of $\sur$ is a {\em Delaunay} triangulation if for each face $f$ of $T$ and any face $\lift{f}$ of $\pb{T}$, there exists an open disk in $\lift{\sur}$ inscribing $\lift{f}$ that is \emph{empty}, i.e., that contains no vertex of $\pb{T}$.
\end{definition}
We will see in Section~\ref{sec:geotri} that any Delaunay triangulation of $\sur$ is geometric.

Remark that, even for a hyperbolic surface, every empty disk in the universal cover $\Hm^2$ is compact. Indeed, any non-compact disk contains at least one disk of any diameter, so, at least one disk of diameter $\diam(T)$, thus it contains a fundamental domain (actually, infinitely many fundamental domains) and cannot be empty. 

Let us now give a natural definition for flips in triangulations of surfaces. It is based on the usual notion of flips in the Euclidean plane. 
\begin{definition}\label{def:flip} Let $T$ be a triangulation of $\sur$. Let $(v_1,v_2,v_3)$ and $(v_2,v_1,v_4)$ be two adjacent triangles in $T$, sharing the edge $e=(v_1,v_2)$. Let us lift the quadrilateral $(v_1,v_2,v_3,v_4)$ to $\lift{\sur}$ so that $(\lift{v_1},\lift{v_2},\lift{v_3})$ and $(\lift{v_2},\lift{v_1},\lift{v_4})$ form two adjacent triangles of $\pb{T}$ sharing the edge $\lift{e}=(\lift{v_1},\lift{v_2})$. 

Flipping $e$ in $T$ consists in replacing the diagonal $\lift{e}$ in the quadrilateral $(\lift{v_1},\lift{v_2},\lift{v_3}, \lift{v_4})$ (which lies in $\lift{\sur}$, i.e., $\Rm^2$ or $\Hm^2$) by the other diagonal $(\lift{v_3},\lift{v_4})$, then projecting the two new triangles $(\lift{v_3},\lift{v_4},\lift{v_2})$ and $(\lift{v_4},\lift{v_3},\lift{v_1})$ to $\sur$ by $\proj$. 

We say that the flip of $T$ along $e$ is {\em Delaunay} if the triangulation is \emph{locally Delaunay} in the quadrilateral after the flip, i.e., the disk inscribing $(\lift{v_3},\lift{v_4},\lift{v_2})$ does not contain $\lift{v_1}$ (and the disk inscribing $(\lift{v_4},\lift{v_3},\lift{v_1})$ does not contain $\lift{v_2}$). 

An edge $e$ is said to be {\em Delaunay flippable} if the flip along $e$ is Delaunay.
\end{definition}
Note that even if $T$ is geometric, the triangulation after a flip is not necessarily geometric.
We will prove later (Lemma~\ref{lm:flip-geom}) that a Delaunay flip transforms a geometric triangulation into a geometric triangulation. 

\paragraph*{Degenerate sets of points on a surface.} Let us quickly examine here the case of degenerate sets of points, i.e., sets of points $V$ on $\sur$ such that the infinite Delaunay triangulation of $\proj^{-1}(V)$ is not unique, i.e., at least two adjacent triangles in the possible Delaunay triangulations of $\proj^{-1}(V)$ in $\lift{\sur}$ have cocircular vertices. 
In such a case, any triangulation of the subset $\mathcal C$ of $\proj^{-1}(V)$ consisting of $c$ cocircular points is a Delaunay triangulation. Any of these triangulations can be transformed in any other by $O(c)$ flips~\cite{ferran}. From now on, we can thus assume that the set of points $V$ on the surfaces that we consider is always non-degenerate.

\paragraph*{Triangulations and polyhedral surfaces.}
The Euclidean plane can be identified with the plane $(z=1)$ in $\R^3$, while the Poincaré model of the hyperbolic plane can be identified with the unit disk in that plane. We can now use the stereographic projection $\sigma:\Sm^2\setminus \{ s_0\}\to \Rm^2$ to send the unit sphere $\Sm^2$ to this plane $(z=1)$, where $s_0=(0,0,-1)$ is the pole. In this projection, each point $p\neq s_0$ on the sphere is sent to the unique intersection with the plane $(z=1)$ of the line going through $s_0$ and $p$. The inverse image of the plane $(z=1)$ is $\Sm^2\setminus \{ s_0\}$, while the inverse image of the disk containing the Poincaré model of the hyperbolic plane is a disk, which is the set of points of $\Sm^2$ above a horizontal plane.

Let $T^{\star}$ be a triangulation of the Euclidean or the hyperbolic plane -- for instance, $T^{\star}$ could be the inverse image $\pb{T}$ of a triangulation $T$ of a surface $(\sur,h)$, in which case $T^{\star}$ has infinitely many vertices. 
We associate to $T^{\star}$ a polyhedral surface $\Sigma$ in $\Rm^3$, constructed as follows. The construction is similar to the classic duality originally presented with a paraboloid in the case of (finite) triangulations in a Euclidean space~\cite{es-vda-86}. It can also be seen as a simpler version, sufficient for our purpose, of the construction presented for triangulations in hyperbolic spaces using the space of spheres~\cite{bdt-hdcvd-14}. 
\begin{itemize}
\item The vertices of $\Sigma$ are the inverse images on $\S^2$ by $\sigma$ of the vertices of $T^{\star}$.
\item The edges of $\Sigma$ are line segments in $\R^3$ corresponding to the edges of $T^{\star}$ and the faces of $\Sigma$ are triangles in $\R^3$ corresponding to the faces of $T^{\star}$. 
\end{itemize}
Note that $\Sigma$ is not necessarily convex. We can make the following well-known remarks. Let $t_1$ and $t_2$ be two triangles of $T^{\star}$ sharing an edge ${e}$, and let $t^\Sigma_1$ and $t^\Sigma_2$ be corresponding faces of the polyhedral surface $\Sigma$, sharing the edge $e^\Sigma$. Then $\Sigma$ is concave at $e^\Sigma$ if and only if $e$ is Delaunay flippable. Flipping $e$ in the triangulation $T^{\star}$ in the plane corresponds to replacing the two faces $t^\Sigma_1$ and $t^\Sigma_2$ of $\Sigma$ by the two other faces of the tetrahedron formed by their vertices. That tetrahedron lies between $\Sigma$ and $\Sm^2$. We obtain a new edge $e^{\Sigma'}$ at which the new polyhedral surface $\Sigma'$ is convex, and which is strictly closer to $\Sm^2$ than $\Sigma$. By an abuse of language, we will say that $\Sigma'$ \emph{contains} $\Sigma$, which we will denote as $\Sigma\subset\Sigma'$.

As a consequence, $\Sigma$ is convex if and only if $T^{\star}$ is Delaunay.

There is a direct corollary of this statement: Given a (non-degenerate, see above) discrete set $V$ of points in $\R^2$ or $\Hm^2$, there is a unique Delaunay triangulation with this set of vertices. 

However we are going to see in the next two sections that there can be infinitely many geometric (non-Delaunay) triangulations on a surface, with the same given finite vertex set.

\section{Geometric triangulations of surfaces} \label{sec:geotri}

We consider now Dehn twists, which are usually considered as acting on the space of metrics on a surface~\cite{casson-bleiler}, but are defined here equivalently, for simplicity, as acting on triangulations of a closed oriented surface $(\sur,h)$ equipped with a fixed Euclidean or hyperbolic structure (figures in this section illustrate the flat case, but the results are proved for both flat and hyperbolic cases). 
Let $T$ be a triangulation of $(\sur,h)$, with vertex set $V$, and let $c$ be an oriented homotopically non-trivial simple closed curve on $\sur\setminus V$. We define a new triangulation $\tau_c(T)$ of $\sur$ by performing a {\em Dehn twist} along $c$: whenever an edge $e$ of $T$ intersects $c$ at a point $p$, we orient $e$ so that the unit vectors of the tangent plane along $e$ and $c$ form a positively oriented basis (see Figure~\ref{fig:dehn_tore} (left)), and then replace $e$ by the oriented path following $e$ until $p$, then following $c$ until it comes back to $p$, then following $e$ until its endpoint (see Figure~\ref{fig:dehn_tore} (right)). This defines a map $\tau_c$ from the space of triangulations of $\Tm^2$ with vertex set $V$ to itself. Note that, even if $T$ is a geometric triangulation, $\tau_c(T)$ is not necessarily geometric. If we denote by $-c$ the curve $c$ with the opposite orientation, then one easily checks that $\tau_{-c}=\tau_c^{-1}$.
\begin{figure}[htb]
    \centering
    \includegraphics{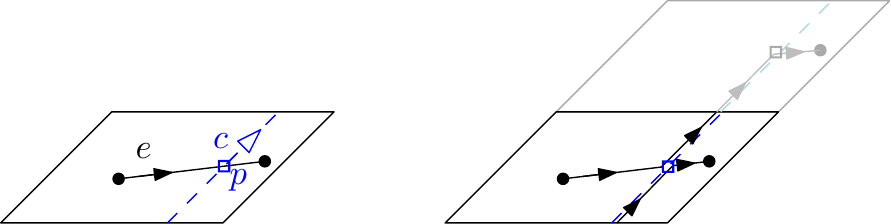}
    \caption{Transformation of an edge $e$ by the Dehn twist along $c$ on a flat torus $\Tm^2$. Here the black parallelepiped is a fundamental domain, and the gray one, used for the construction of the image of $e$ by $\tau_c$, is another fundamental domain, image through an element of the the group  $\Gamma$ of isometries.}
    \label{fig:dehn_tore}
\end{figure}

\begin{lemma}
  There exists a geometric triangulation $T$ of $(\sur,h)$ and a simple closed curve $c\subset \sur$ such that for all $k\in \Z$, $\tau_c^k(T)$ is geometric.
\end{lemma}

\begin{proof}
We choose a simple closed geodesic $c$ on $(\sur,h)$ and $\varepsilon>0$. We denote by $c_-, c_+$ the two geodesics at distance $\varepsilon$ from $c$ on the positive and negative sides of $c$. The value of $\varepsilon$ must be sufficiently small so that the region between $c_-$ and $c_+$ is an annulus drawn on $\sur$. We then choose a geometric triangulation $T$ of $(\sur,h)$ with no vertex in the open annulus bounded by $c_-$ and $c_+$ and containing $c$, such that each edge crossing $c$ intersects $c$ exactly once, and has one endpoint on $c_-$ and another on $c_+$.

We realize the image by $\tau_c$ of an edge $e$ of $T$ as a geodesic segment -- there is a unique choice in the homotopy class of the path described above (Figure~\ref{fig:dehn_tore_geodesic}).
\begin{figure}[htb]
    \centering
    \includegraphics{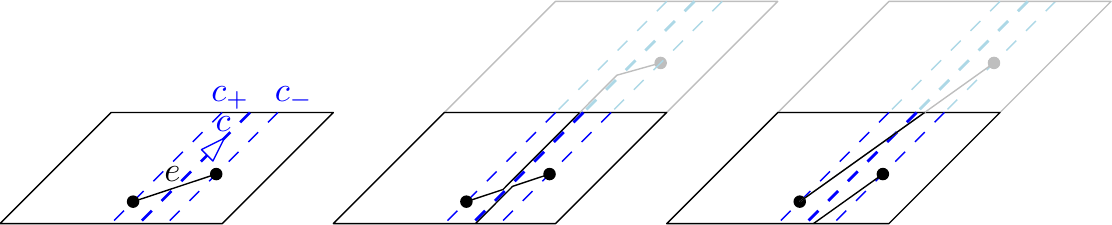}
    \caption{Image of $e$ by a Dehn twist (middle), realized as a geodesic edge (right).}
    \label{fig:dehn_tore_geodesic} 
\end{figure} 
Let $e,e'$ be two edges of $T$. If either $e$ or $e'$ does not intersect $c$, then their images by $\tau_c$ (or $\tau_{-c}$) remain disjoint, as they lie in different regions separated by $c_-$ and $c_+$.  If $e$ and $e'$ intersect $c$, then again their images by $\tau_c$ (or $\tau_{-c}$) remain disjoint, as their endpoints appear in the same order on $c_-$ and $c_+$ and two geodesic lines cannot intersect more than once (Figure~\ref{fig:dehn_tore_two}).  
\begin{figure}[htb]
    \centering
    \includegraphics{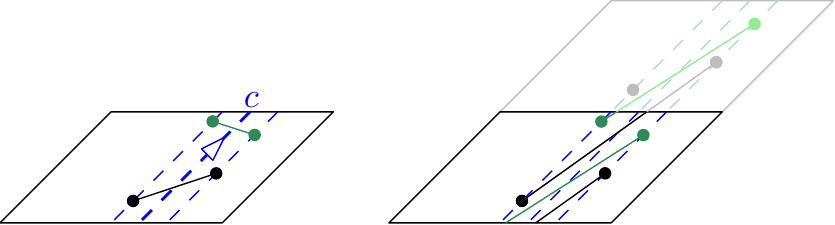}
    \caption{The Dehn twist of two edges along $c$ for two edges intersecting $c$.}
    \label{fig:dehn_tore_two}
\end{figure}
As a consequence, $\tau_c(T)$ (and $\tau_{-c}(T)$) are geometric. The same result follows by induction for $\tau_c^{k}(T)$ for any $k\in \Z$. 
\end{proof}

\begin{corollary}\label{cor:infinite}
  For any closed oriented surface $(\sur,h)$, there exists a finite set of points $V\subset\sur$ such that the graph of geometric triangulations with vertex set $V$ is infinite.
\end{corollary}

We can now prove the following result:
\begin{proposition} \label{pr:geometric}
  Any Delaunay triangulation of a closed oriented surface $(\sur,h)$ is geometric.
\end{proposition}

\begin{proof}
  Let $V$ be a finite set of points on $\sur$, and let $T$ be the Delaunay triangulation of $(\sur,h)$ with vertex set $V$. Realize every edge of $T$ as a the unique geodesic segment in its homotopy class. We argue by contradiction and suppose that $T$ is not geometric, so that there are two edges $e_1$ and $e_2$ that intersect in their interiors. We then lift $e_1$ and $e_2$ to edges $\lift{e_1}$ and $\lift{e_2}$ of $\pb{T}$ whose interiors still intersect at one point. 

There are at least two distincts faces $\lift{f_1}$ and $\lift{f_2}$ of $\pb{T}$ such that $\lift{e_1}$ is an edge of $\lift{f_1}$ and $\lift{e_2}$ is an edge of $\lift{f_2}$. Let $\lift{C_1}$ and $\lift{C_2}$ be the circles inscribing $\lift{f_1}$ and $\lift{f_2}$, respectively. Since $\pb{T}$ is Delaunay, $\lift{C_1}$ and $\lift{C_2}$ bound empty disks $\lift{D_1}$ and $\lift{D_2}$, i.e., open disks not containing any point of $\proj^{-1}(V)$. Recall that, as mentioned in Section~\ref{sec:delaunay}, empty disks are compact even in the hyperbolic case, and that $\lift{e_1}\subset\lift{D_1}$ and $\lift{e_2}\subset\lift{D_2}$ (edges are considered as open).

The two circles $\lift{C_1}$ and $\lift{C_2}$ do not intersect more than twice. Let $\lift{L}$ be the geodesic line through their two intersection points. The endpoints of $\lift{e_1}$ are on $\lift{C_1}\setminus\lift{D_2}$ and those of $\lift{e_2}$ are on $\lift{C_2}\setminus\lift{D_1}$, so the two pairs of endpoints are on opposite sides of $\lift{L}$. As a consequence, $\lift{e_1}$ and $\lift{e_2}$ are on opposite sides of $\lift{L}$, so they cannot intersect. This leads to a contradiction.
\end{proof}

\section{The flip algorithm}
Let us consider a closed oriented surface $(\sur,h)$. 
The flip algorithm consists in performing Delaunay flips in any order,
starting from a given input geometric triangulation of $\sur$, until there is no
more Delaunay flippable edge. 

In this section, we first define a data structure that supports this
algorithm, then we prove the correctness of the algorithm. 

\subsection{Data structure}
In both cases of
a flat or hyperbolic surface, the group of isometries defining the
surface is denoted as $\gp$. 
We assume that a fundamental
domain $\F$ is given. By definition (Section~\ref{sec:surface}),
$\lift{\sur}$ is the union $\gp(\F)$ of the images of $\F$ under
the action of $\gp$. 

To represent a triangulation on the surface, we propose a data
structure generalizing the data structure previously introduced for
triangulations of flat orbifolds \cite{ct-dtced-16} and triangulations
of the Bolza surface \cite{it-idtbs-17}. The combinatorics of the
triangulation is given by the set of its vertices $V$ on the surface
and the set of its triangles, where each triangle gives access to its
three vertices in $V$ and its three adjacent triangles, and each
vertex gives access to one of its incident triangles. The geometry of
the triangulation is given by the set $\rep{V}$ of the
lifts of its vertices that lie in
the fundamental domain $\F$ and one lift $\rep{t}$ in $\lift{\sur}$ of
each triangle $t=(v_{0,t};v_{1,t};v_{2,t})$ of the triangulation,
chosen among the (one, two, or three) lifts of $t$ in $\lift{\sur}$
having at least one
vertex in $\F$: $\rep{t}$ has at least one of its vertices
$\rep{v_{i,t}}$ in $\F$ ($i=0,1$, or $2$); then the other vertices of
$\rep{t}$ are images $\is_{i+1,t}\cdot\rep{v_{i+1,t}}$ and
$\is_{i+2,t}\cdot\rep{v_{i+2,t}}$ of two vertices in $\rep{V}$, where $\is_{i+1,t}$ and
$\is_{i+2,t}$ are elements of $\gp$ (indices are taken
modulo~3). In the data structure, each vertex $v$ on the surface has
access to its representative $\rep{v}$, and each triangle $t$ on the
surface has access to the isometries $\is_{0,t},\is_{1,t}$, and
$\is_{2,t}$ allowing to construct $\rep{t}$, at least one of the isometries
being the identity $\e$. 
Note that two triangles $t$ and
$t'$ of $T$ that are adjacent on the surface are represented by two
triangles $\rep{t}$ and $\rep{t'}$, which are not necessarily adjacent
in $\lift{\sur}$ (Figure~\ref{fig:notation-flip} (left)). However,
there is an isometry $\is$ in $\gp$ such that 
$\rep{t}$ and $\is\cdot\rep{t'}$ are adjacent.

Let $T$ be an input triangulation given as such a data
structure.
Figure~\ref{fig:notation-flip} illustrates a Delaunay flip performed
on two adjacent triangles $t$ and $t'$ on the surface. The triangle
$\rep{t'}$ is first moved so that the vertices of the edge to be
flipped coincide. Then the edge is flipped. The isometries in the two
triangles created by the flip are easy to compute from the isometries
stored in $t$ and $t'$.
Note that the order in which isometries are composed is crucial in
the hyperbolic case, as they do not commute. We have shown that the
data structure can be maintained through flips. 

\begin{figure}[htb]
    \centering
    \includegraphics{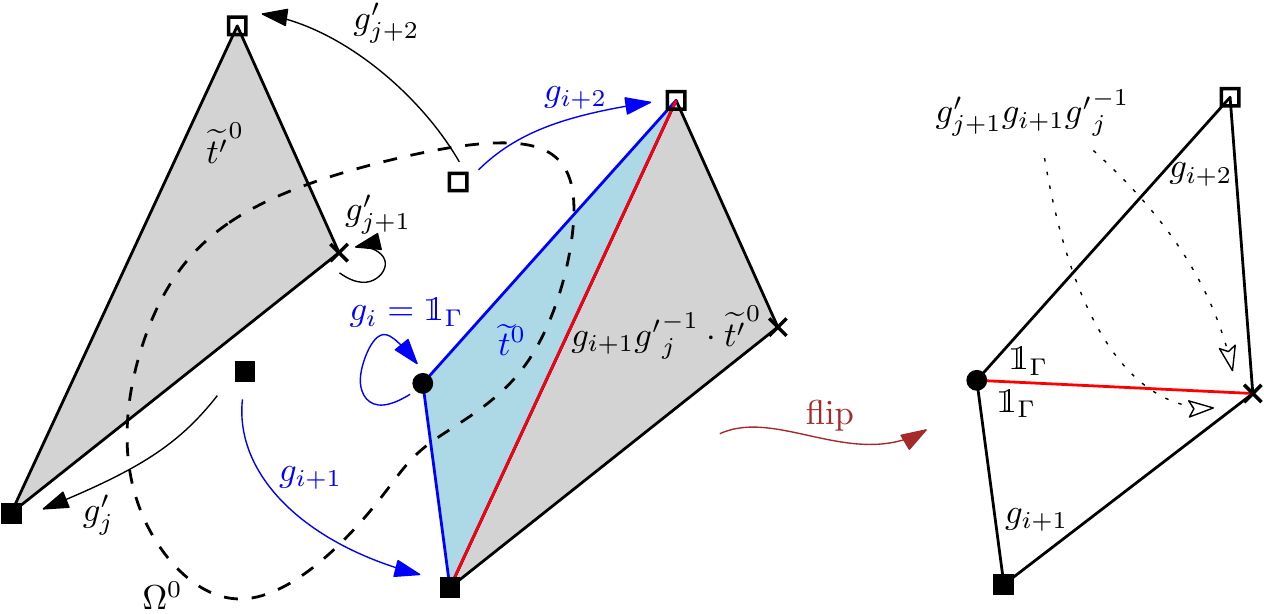}
    \caption{A flip. Here (hyperbolic) triangles are represented schematically with
      straight edges. Left: the two triangles $\rep{t}$ and
      $\rep{t'}$ before the flip. Here $\is_i=\e$. Right: the
      isometries in the two triangles created by the flip.}
    \label{fig:notation-flip}
\end{figure}

\subsection{Correctness of the algorithm}
The following statement is a key starting point. 

\begin{lemma}\label{lm:flip-geom}
Let $T$ be a geometric triangulation of $(\sur,h)$, and let $T'$ be obtained
from $T$ by a Delaunay flip. Then $T'$ is still
geometric.
\end{lemma}

\begin{proof}
Let $e$ be a Delaunay flippable edge and $\lift{e}$ a lift
in $\lift{\sur}$. 
Denote the vertices of $\lift{e}$ by $\lift{v}$ and $\lift{v'}$. Let $\lift{t_1}$
and $\lift{t_2}$ be the triangles of $\pb{T}$ incident to
$\lift{e}$. To prove that $T'$ is geometric, it is sufficient to prove
that $\lift{t_1}\cup\lift{t_2}$ is a strictly convex quadrilateral.

Let $\lift{C_1}$ (resp. $\lift{C_2}$) be the circle through the
three vertices of $\lift{t_1}$ (resp.\ $\lift{t_2}$). Note that
$\lift{C_1}$ and $\lift{C_2}$ may be non-compact. Let $\lift{D_1}$ and
$\lift{D_2}$ be the corresponding disks (as defined in
Section~\ref{sec:hyperb} on case of non-compact circles). The disk
$\lift{D_1}$ (resp.\ $\lift{D_2}$) is convex (in the Euclidean plane
if $\sur$ is a flat torus, or in the sense of hyperbolic geometry if
$\sur$ is a hyperbolic surface) and contains $\lift{t_1}$ (resp.\
$\lift{t_2}$). The fact that $e$ is Delaunay flippable then implies
that $\lift{t_1}$ and $\lift{t_2}$ are contained in
$\lift{D_1}\cap\lift{D_2}$ (see Figure~\ref{fig:flip-geometric}).
\begin{figure}[htb]
    \centering
    \includegraphics{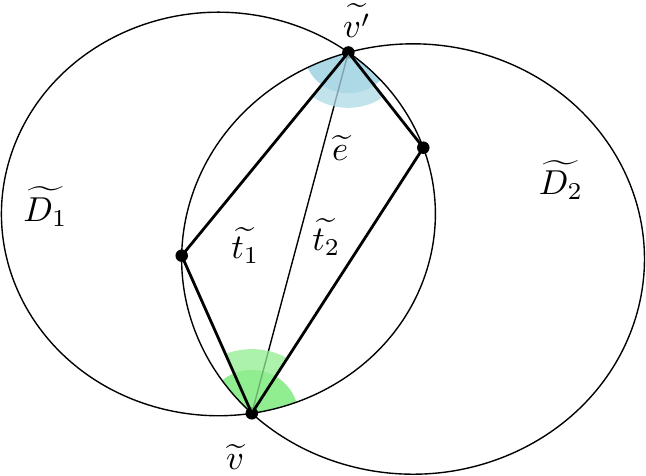}
    \caption{The quadrilateral is convex (edges are represented
      schematicaly as straight line segments).}
    \label{fig:flip-geometric}
\end{figure}
As a consequence, the sum of angles of $\lift{t_1}$ and $\lift{t_2}$ at $\lift{v}$ is smaller
than the interior angle at $\lift{v}$ of $\lift{D_1}\cap\lift{D_2}$, which is at most
$\pi$, and similarly at $\lift{v'}$. As a consequence, the quadrilateral
$\lift{t_1}\cup\lift{t_2}$ is strictly convex at $\lift{v}$ and $\lift{v'}$. Since it is strictly
convex at its other two vertices (as each of these vertices is a
vertex of a triangle), it is strictly convex, and the statement follows.
\end{proof}

The following lemma, using the diameter of the triangulation
(Definition~\ref{def:diam}), is central in the proof of the
termination of the algorithm (Theorem~\ref{th:flip-hyp}) for hyperbolic surfaces and in its
analysis for both flat tori and hyperbolic surfaces (Section~\ref{sec:analysis}).

\begin{lemma}\label{lem:edgelength}
	Let $T$ be a geometric triangulation of $(\sur,h)$. Then, the flip algorithm starting from $T$ will never insert an edge longer than $2\diam(T)$.
\end{lemma}
Note that the length of an edge can be measured on any or its lifts in the universal
covering space $\lift{\sur}$.
\begin{proof}
  Let $T_k$ be the triangulation obtained from $T=T_0$ after $k$ flips
  and let ${\Sigma_k}$ be the corresponding polyhedral surface of
  $\Rm^3$ as defined in Section~\ref{sec:delaunay}. Since we perform
  only Delaunay flips,
  ${\Sigma_0}\subset\ldots\subset{\Sigma_k}\subset{\Sigma_{k+1}}$ (with the
  abuse of language mentioned in Section~\ref{sec:delaunay}).

We will prove the result by contradiction. Let us
  assume that $T_k$ has an edge $e$ of length larger than
  $2\diam(T)$. Let $\Omega$ be a fundamental domain of $\sur$ having
  diameter $\diam(T)$, given as the union of lifts of triangles of $T=T_0$
  (it is not clear how to compute such a fundamental domain
  efficiently but its existence is clear). Let $v$ be the midpoint of
  $e$ and $\lift{v}$ its lift in $\Omega$. Let
  $\lift{e}=(\lift{v_1},\lift{v_2})$ be the unique lift of $e$ whose
  midpoint is $\lift{v}$. The domain $\Omega$ is strictly included in the
  disk $\lift{D}$ of radius
  $\diam(T)$ and centered at $\lift{v}$, by definition of $\diam(T)$
  (see Figure~\ref{fig:no-long-edge} (left)).

  Let $P_D$ denote the plane in $\R^3$ containing the circle on $\Sm^2$ that is the boundary 
  of $\sigma^{-1}(D)$ (recall that $\sigma$
  denotes the stereographic projection, see
  Section~\ref{sec:delaunay}), and let $p$ denote the point
  $\sigma^{-1}(\lift{v})$ on $\Sm^2$. As
  $p\in \sigma^{-1}(\Omega)\subset\sigma^{-1}(\lift{D})$, the projection
  $p^{\Sigma_0}$ of $p$ onto $\Sigma_0$ lies above $P_D$
  (Figure~\ref{fig:no-long-edge} (right)).

\begin{figure}[htb]
    \centering
    \includegraphics{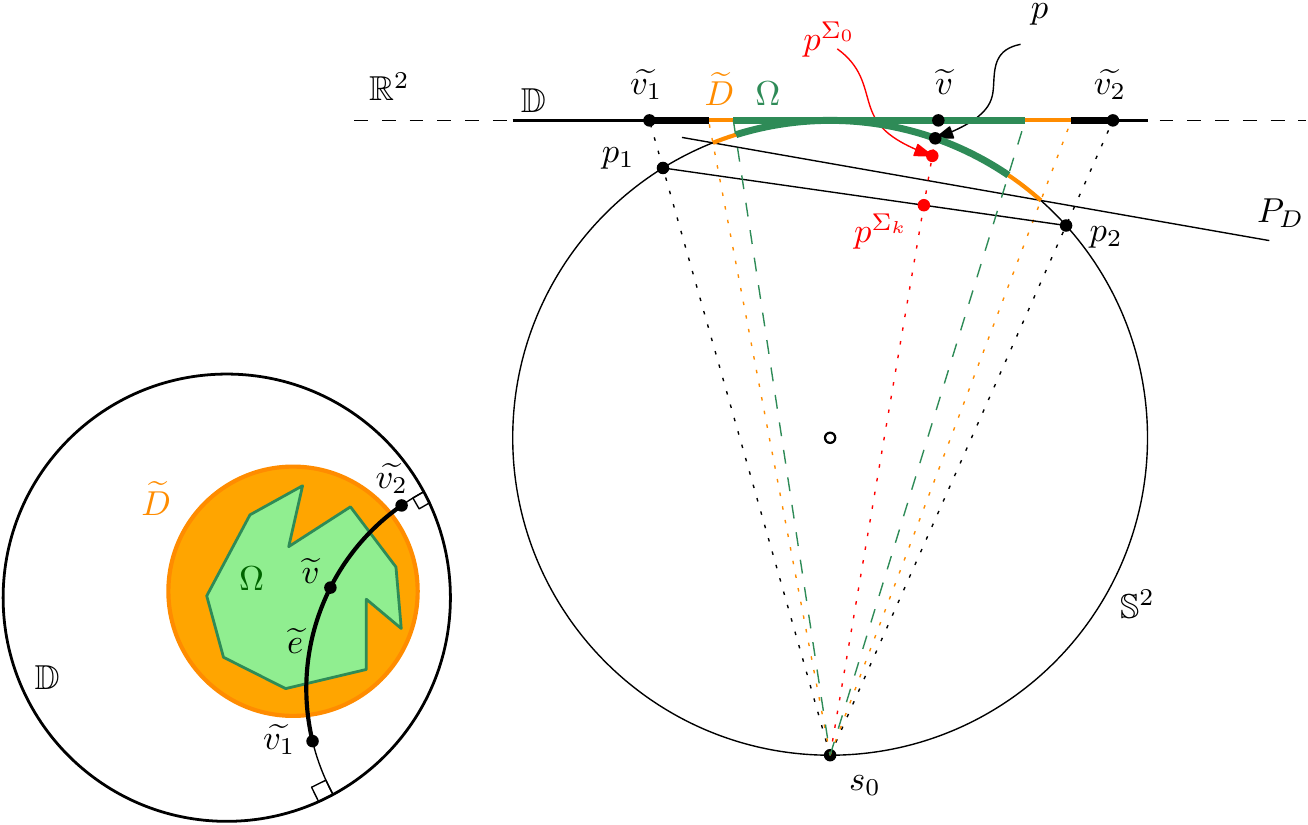}
    \caption{Illustration for the proof of Lemma~\ref{lem:edgelength}
      (for a hyperbolic surface). Left: notation in
      $\Hm^2$. Right: contradiction seen in a cutting plane in $\R^3$.}
    \label{fig:no-long-edge}
\end{figure}

Now, denote the edge $\sigma^{-1}(\lift{e})$ on $\Sm^2$ as $(p_1,p_2)$. The points $p_1$ and $p_2$
lie outside $\sigma^{-1}(D)$. So, the corresponding edge
$e^{\Sigma}=[p_1,p_2]$ of $\Sigma_k$ lies below the plane $P_C$, thus
the projection $p^{\Sigma_k}\in [p_1,p_2]$ of $p$ onto $\Sigma_k$ lies below $P_C$.  

From what we have shown, $p^{\Sigma_k}$ is a point of $\Sigma_k$ that
lies strictly between the pole $s_0$ and the point $p^{\Sigma_0}$ of
$\Sigma_0$, which contradicts the inclusion $\Sigma_0\subset\Sigma_k$. 
\end{proof}

We will now show that, for any order, the flip algorithm terminates
and returns the Delaunay triangulation of the surface. The proof given
for the hyperbolic case would also work for the flat case. However we
propose a more elementary proof for the flat case.

\paragraph*{Flat tori}

The case of flat tori is easy, and might be considered as folklore. However, as we have not found a reference, we give the details here for completeness. 

We define the weight of a triangle $t$ of a geometric triangulation $T$ of $\Tm^2$ as the number of vertices of $\pb{T}$ that lie in the open circumdisk of a lift of $t$. The weight $w(T)$ of $T$ is defined as the sum of the weights of its triangles. 

\begin{lemma}
The weight $w(T)$ of a triangulation $T$ of a flat torus $(\Tm^2, h)$ is finite. Let $T'$ be the triangulation obtained from a geometric triangulation $T$ after performing a Delaunay flip. Then $w(T') \leq w(T) - 2$. 
\end{lemma}

\begin{proof}
A circumdisk of any triangle in $\R^2$ is compact, so, it can only contain a finite number of vertices of $\pb{T}$. The sum $w(T)$ of these numbers over triangles of $T$ is clearly finite as the number of triangles of $T$ is finite. Let us now focus on a quadrilateral in $\R^2$ that is a lift of the quadrilateral on $\Tm^2$ whose diagonal $e$ is flipped. Let $\lift{D_1}$ and $\lift{D_2}$ denote the two open circumdisks in $\R^2$ before the flip and $\lift{D'_1}$ and $\lift{D'_2}$ denote the two open circumdisks after the flip, then $\lift{D'_1}\cup\lift{D'_2} \subset \lift{D_1}\cup\lift{D_2}$ and $\lift{D'_1}\cap\lift{D'_2} \subset \lift{D_1}\cap\lift{D_2}$ (see Figure~\ref{fig:disks-tore}).
\begin{figure}[htb]
    \centering
    \includegraphics{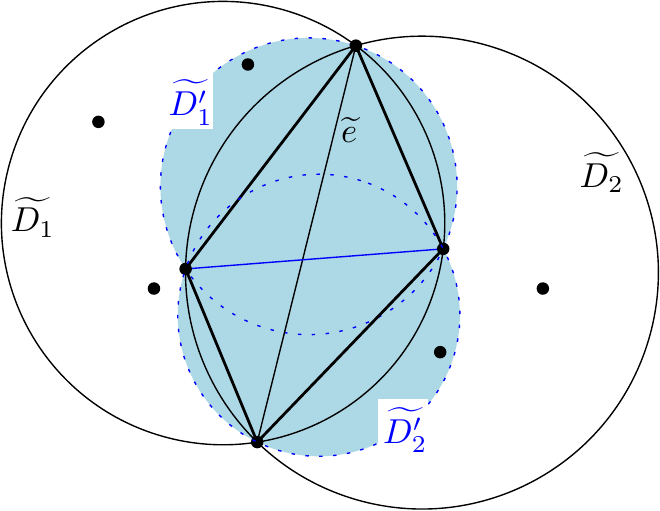}
    \caption{Circumdisks $\lift{D_1}$ and $\lift{D_2}$ before flipping $\lift{e}$ and $\lift{D_1}$ and $\lift{D'_2}$ after the Delaunay flip.}
    \label{fig:disks-tore}
\end{figure}
 Moreover, by definition of a Delaunay flip, the union $\lift{D'_1}\cup\lift{D'_2}$ contains at least two fewer vertices of $\pb{T}$ than $\lift{D_1}\cup\lift{D_2}$, which are the two vertices of the quadrilateral that are not vertices of $\lift{e}$. This concludes the proof. 
\end{proof}

The result follows trivially:

\begin{theorem}\label{th:flip-torus}
Let $T$ be a geometric triangulation of a flat torus with finite vertex set $V$. The flip algorithm terminates and outputs the Delaunay triangulation of $V$. 
\end{theorem}

\begin{corollary}
	The geometric flip graph $\cF_{\Tm^2, h,V}$ is connected.
\end{corollary}

\paragraph*{Hyperbolic surfaces}%\label{sec:hypDT}
To show that the flip algorith terminates in the hyperbolic case, we cannot mimic the proof presented for the flat tori since the circumcircle of a hyperbolic triangle can be non-compact (see Section~\ref{sec:hyperb}) and thus can have an infinite weight. Note also that the proof cannot use a property on the angles of the Delaunay triangulation similar to what holds in the Euclidean case: in $\Hm^2$, the locus of points seeing a segment with a given angle is not a circle arc, and thus the Delaunay triangulation of a set of points in $\Hm^2$ does not maximize the smallest angle of triangles.
The proof relies on Lemma~\ref{lem:edgelength}.

\begin{theorem}\label{th:flip-hyp}
	Let $T$ be a geometric triangulation of a closed hyperbolic surface with finite vertex set $V$. The flip algorithm terminates and outputs the Delaunay triangulation of~$V$.
\end{theorem}

\begin{proof}
We use the same notation as in the proof Lemma~\ref{lem:edgelength}. Once an edge of $T_k$ is flipped, it can never reappear in the triangulation, as the corresponding segment in $\R^3$ becomes interior to the polyhedral surface $\Sigma_{k+1}$ (see Section~\ref{sec:delaunay}) and further surfaces $\Sigma_{k'}, k'\geq k+1$. In addition, all the introduced edges have length smaller than $2\diam(T)$ by Lemma~\ref{lem:edgelength}. Moreover, there is only a finite number of edges with vertices in $V$ that are shorter than $2\diam(T)$ on $S$, as a circle given by a center and a bounded radius is compact. So, the flip algorithm terminates. The output does not have any Delaunay flippable edge, so, it is the Delaunay triangulation.  
\end{proof}

\begin{corollary}
	The geometric flip graph $\cF_{S, h,V}$ is connected.
\end{corollary}

\section{Algorithm analysis}\label{sec:analysis}
For a triangulation on $n$ vertices in the Euclidean plane, counting the weights of triangulations leads to the optimal $O(n^2)$ bound. However the same argument does not yield a bound even for the flat torus, since points must be counted in the universal cover. 

\begin{theorem} \label{tm:bound-torus}
	For any triangulation $T$ with $n$ vertices of a torus $(\Tm^2,h)$, there is a sequence of flips of length $C_h\cdot\diam(T)^2\cdot n^2$ connecting $T$ to a Delaunay triangulation of $(\Tm^2,h)$, where $C_h$ only depends on $h$.
\end{theorem}

\begin{proof}
  Let $e=(v_1,v_2)$ be an edge appearing during the flip algorithm, and $\lift{v_1}$ (resp.\ $\lift{v_2}$) be a lift of $v_1$ (resp.\ $v_2$), such that $(\lift{v_1},\lift{v_2})$ is a lift $\lift{e}$ of $e$. The point $\lift{v_2}$ lies in a circle $C$ of diameter $4\diam(T)$ centered at $\lift{v_1}$ by Lemma~\ref{lem:edgelength}. Let $M$ be the affine transformation that maps the lattice of the lifts of $v_2$ to the square lattice $\mathbb{Z}^2$. $M(C)$ is a convex set and from Pick's theorem~\cite{trainin},\footnote{See also \url{https://en.wikipedia.org/wiki/Pick's_theorem\#Inequality_for_convex_sets}} the number of points of $\mathbb{Z}^2$ in $M(C)$ is smaller than $\mbox{area}(M(C)) + 1/2\cdot\mbox{perimeter}(M(C))+1$, which is also a bound on the number of possible points $\lift{v_2}$ in $C$ and thus the number of possible edges $e$. The area of $M(C)$ is $1/A_h\cdot\mbox{area(C)}$ since $det(M)=1/A_h$, but there is no simple formula for its perimeter. As already mentioned in the proof of Theorem~\ref{th:flip-hyp}, an edge can never reappear after it was flipped. Moreover, there are $n^2/2$ sets of points $\{v_1,v_2\}$ ($v_1$ and $v_2$ may be the same point), which yields the result.
\end{proof}

The rest of this section is devoted to computing the number of edges not longer than $2\diam(T)$ between two fixed points $v_1$ and $v_2$ on a hyperbolic surface $(S,h)$. Counting the number of points in a disk of fixed radius would give an exponential bound because the area of a circle in $\Hm^2$ is exponential in its radius~\cite{m-aetim-69}. Note that we only consider geodesic edges, so we only need to count homotopy classes of simple paths. The behavior of the number $N_l$ of simple closed curves smaller than a fixed length $l$ is well understood: $N_l/l^{6g-6}$ converges to a positive constant depending ``continuously'' on $h$~\cite{m-gnscg-08}. However, we need a result for geodesic paths instead of geodesic closed curves, and Mirzakhani's proof is too deep and relies on too sophisticated structures to easily be generalized. So, we will only prove an upper bound on the number of paths. Such an upper bound could be derived from the theory of measured laminations of Thurston, which is also quite intricate. Fortunately, a more comprehensible proof, specific to simple closed geodesic curves on hyperbolic structures, can be found in a book published by the French Mathematical Society~\cite[4.III, p.61-67]{flp-tts-79}\cite{flp-tws-12}. While recalling the main steps of the proof, we show how to extend it to geodesic paths.

Let $\Gamma=\{\gamma, i=1,\ldots,3g-3\}$ be a set of $3g-3$ simple disjoint closed geodesics on $(S,h)$ not containing $v_1$ and $v_2$ that forms a pants decomposition on $S$, where each $\gamma_i$ belongs to two different pairs of pants. A set  $\{\overline{\gamma_i}, i=1,\ldots,3g-3\}$ of disjoint closed annuli is defined on $S$, where each $\overline{\gamma_i}$ is a tubular neighborhood of $\gamma_i$ containing none of $v_1,v_2$. This yields a decomposition of $S$ into $3g-3$ annuli $\overline{\gamma_i} (i=1,\ldots,3g-3)$ and $2g-2$ pairs of  ``short pants'' $P_j (j=1,\ldots,2g-2)$. For $i=1,\ldots,3g-3$, let us denote as $\partial\overline{\gamma_i}$ any one of the two curves bounding the annulus $\overline{\gamma_i}$ (this is an abuse of notation but should not introduce any confusion).
In each pair of pants $P_j, j=1,\ldots,2g-2$, for each boundary $\partial\overline{\gamma}$, an arc $J^\gamma_i$ is drawn in $P_i$, going from the boundary of $\overline{\gamma}$ to itself that separates the other two boundaries of $P_i$ and that has minimal length. 

Two curves $\gamma'$ and $\gamma''$ are associated to each $\gamma\in\Gamma$ in the following way (Figure~\ref{fig:pantalons}). The annulus $\overline{\gamma}$ is glued with the two pairs of pants $P_i$ and $P_j$ between which it is lying, which yields a sphere with four boundaries: $\partial\overline{\gamma_{i,1}}$ and $\partial\overline{\gamma_{i,2}}$ bounding  $P_i$ and $\partial\overline{\gamma_{j,1}}$ and $\partial\overline{\gamma_{j,2}}$ bounding $P_j$.
\begin{figure}[htb]
    \centering
    \includegraphics{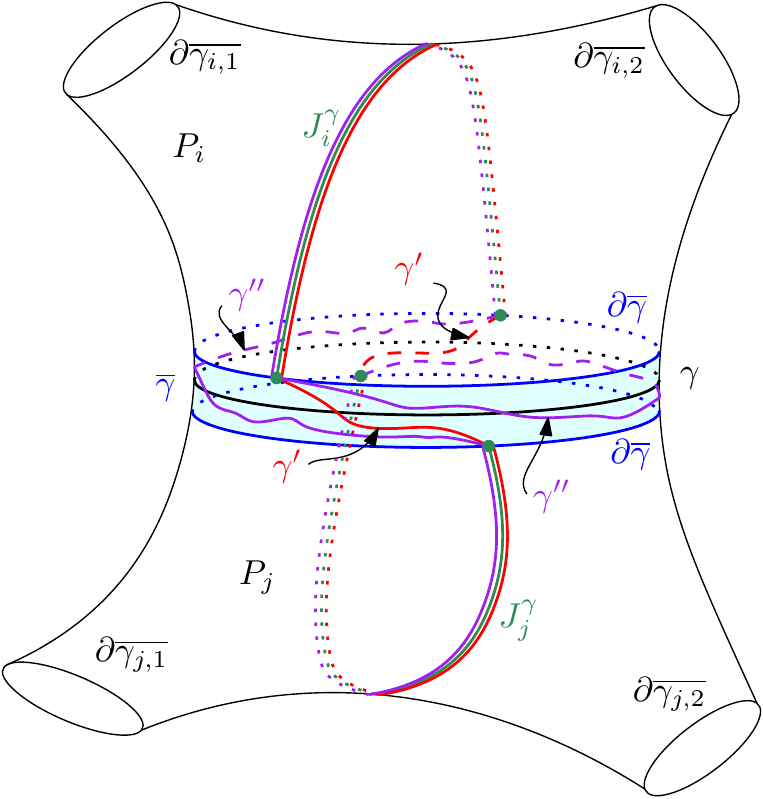}
    \caption{Two adjacent pairs of pants $P_i$ and $P_j$.}
    \label{fig:pantalons}
\end{figure}
A curve $\gamma'$ is then defined: it coincides with $J^\gamma_i$ in $P_i$ and $J^\gamma_j$ in $P_j$, it separates $\partial\overline{\gamma_{i,1}}$ and $\partial\overline{\gamma_{j,1}}$ from $\partial\overline{\gamma_{i,2}}$ and $\partial\overline{\gamma_{j,2}}$, and it has exactly 2 crossings with $\gamma$. The curve $\gamma''$ is defined in the same way, separating $\partial\overline{\gamma_{i,1}}$ and $\partial\overline{\gamma_{j,2}}$ from $\partial\overline{\gamma_{i,2}}$ and $\partial\overline{\gamma_{j,1}}$. 

For each $P_i$ and $m_{i,1},m_{i,2},m_{i,3}\in\mathbb{N}$, a model multiarc is fixed in $P_i$, having $m_{i,1}$, $m_{i,2}$ and $m_{i,3}$ intersections with the three boundaries $\partial\overline{\gamma_{i,1}}$, $\partial\overline{\gamma_{i,2}}$ and $\partial\overline{\gamma_{i,3}}$ of $P_i$ (if one exists). The model is chosen among all the possible model multiarcs as the one that has a minimal number of intersections with the three arcs $J^{\gamma_{i,j}}_i (j=1,2,3)$ of $P_i$.
The model multiarcs is unique, up to homeomorphisms of the pair of pants, and those homeomorphisms are rather simple to understand since they can be decomposed into three Dehn twists around curves homotopic to the three boundaries of the pair of pants. 

Let now $f$ be a path between $v_1$ and $v_2$ on $S$. We decompose $f$ into three parts: $(v_1,w_1)$, $(w_1,w_2)$ and $(w_2,v_2)$ where $w_1$ and $w_2$ are the first and the last point of $f$ on an annulus boundary. 
We ``push'' all the twists of $f^w$ into the annuli $\overline{\gamma}, \gamma\in\Gamma$, and obtain a \emph{normal form} homotopic to $f$, whose definition adapts the definition given in the book~\cite{flp-tts-79}
for closed curves: 
\begin{enumerate}
	\item It is simple.
	\item It has a minimal number $m_i$ of intersections with each $\gamma_i, i=1,\ldots,3g-3$.
	\item In each $P_j, j=1,\ldots,2g-2$, it is homotopic with fixed endpoints to the model multiarc that corresponds to the number of intersections with its boundaries. For $P_{j_1}$  (resp.\ $P_{j_2}$) containing $v_1$ (resp.\ $v_2$), only the intersections different from $w_1$ (resp.\ $w_2$) are counted.
	\item Between $v_1$ and $w_1$ (resp.\ $w_2$ and $v_2$), it has a minimal number of intersections with the three arcs $J^{\gamma_{j_1,k}}_{j_1} (k=1,2,3)$ in $P_{j_1}$ containing $v_1$ (resp.\ $J^{\gamma_{j_2,k}}_{j_2}$ in $P_{j_2}$ containing $v_2$).
	\item It has a minimal number $t_i$ of intersections with $\gamma'_i$ inside $\overline{\gamma_i}$, for any $i=1,\ldots,3g-3$.
	\item It has a minimal number $s_i$ of intersections with $\gamma''_i$ inside $\overline{\gamma_i}$, for any $i=1,\ldots,3g-3$.
\end{enumerate}  
The existence of a normal form is clear but its uniqueness is unclear (uniqueness is not required for the upper bound that we are looking for, but it can actually be proved by extension of the next lemma). 
The two forms of the path $f$ are used to define two notions of complexity: its geodesic form is used to define its length, which can be seen as a geometric complexity, whereas its its \emph{normal coordinates} $m_i$, $s_i$ and $t_i$ can be seen as a combinatorial complexity. Lemma~\ref{lem:pathlength} shows some equivalence between the two notions of complexity. We first show that a fixed set of coordinates corresponds to a finite number of possible non-homotopic paths.

\begin{lemma}\label{lem:paths}
For any set of coordinates $m_i, t_i, s_i, i=1,\ldots,3g-3$, there are at most $9(\max_{\{i=1,\ldots,3g-3\}}(m_i))^2$ non-homotopic normal forms.
\end{lemma}
\begin{proof} Let $f$ be a path, decomposed as above into $(v_1,w_1)$, $f^w=(w_1,w_2)$ and $(w_2,v_2)$. The uniqueness for closed curves comes from the facts that in each pair of pants, fixing the $m_i$, $s_i$ and $t_i$ leads to a unique homotopy class of model multiarcs~\cite[Lemma 5, p.63]{flp-tts-79}. 
Everything remains true but the uniqueness of the homotopy class of model multiarcs in the two (not necessarily different) pairs of pants $P_{j_1}$ and $P_{j_2}$ containing $v_1$ and $v_2$. However, $w_1$ and $w_2$ are fixing unique models (see Figure~\ref{fig:pathVScurve}). 
	\begin{figure}[htb]
		\centering
		\def\svgwidth{\columnwidth}
		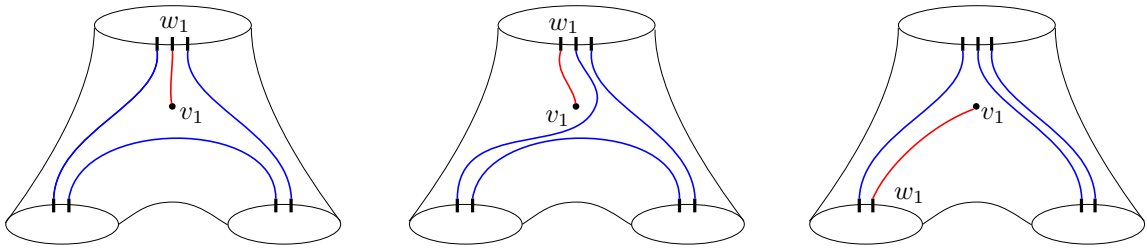
		\caption{Three possible choices for $w_1$. The two left choices correspond to the same model multiarcs, but the orderings on the upper boundary lead to non-homotopic paths. The right choice leads to different model multiarcs.}
		\label{fig:pathVScurve}
	\end{figure}
	There are three possible annulus boundaries $\partial\overline{\gamma_{j,i}}, i=1,2,3$ for $w_1$ in the pair of pants $P_j$ that contains $v_1$ (resp.\ $\overline{\gamma_{j,i}}$ for $w_2$), so, at most $3\max_{\{i\}}(m_i)$ possibilities for each of them. The choices for $w_1$ and $w_2$ are independent and the result follows.
\end{proof}

\begin{lemma}\label{lem:pathlength}
	Let $f$ be a geodesic path of length $l$, then there exists a constant $c_h$ such that the coordinates $m_i, t_i$ and $s_i, i=1,\ldots,3g-3$ of the normal form of $f$ are smaller than $c_h\cdot l$.
\end{lemma}
\begin{proof}
	For any simple closed geodesic $\delta$ on $S$, the geodesic form of $f$ intersects $\delta$ in a minimal number $k_\delta$ of points, since they are both geodesics. If $\varepsilon_\delta$ is the width of a tubular neighborhood of $\delta$, 
then $l\ge \varepsilon_\delta(k_\delta-1)$~\cite[Lemma 3.1]{bs-1985}. Each coordinate $m_i, t_i$ and $s_i$ of $f$ corresponds to the minimal number of intersections with a curve. The number $m_i$ corresponds to $\gamma_i$. The number $t_i$ is actually not larger 
than the number of intersections of $f$ with the geodesic curve that is homotopic to $\gamma'_i$ ($\gamma'_i$ is generally not geodesic), and similarly $s_i$ is not larger than the number of intersections of $f$ with the geodesic homotopic to $\gamma''_i$. 
These curves $\gamma_i,\gamma'_i,\gamma''_i$ only depend on $(S,h)$, so, we can take $\varepsilon_h$ to be the largest of all the $9g-9$ widths  $\varepsilon_{\gamma_i}, \varepsilon_{\gamma'_i}, \varepsilon_{\gamma''_i}$  and we obtain $l\ge\varepsilon_h\cdot\max(m_i,t_i,s_i)$ and thus $\max(m_i,t_i,s_i)\le1/\varepsilon_h\cdot l$.
\end{proof}

\begin{theorem} \label{tm:bound-hyp}
  For any hyperbolic structure $h$ on $S$ and any triangulation $T$ of $(S,h)$, there is a sequence of flips of length at most $C_h\cdot \diam(T)^{6g-4}\cdot n^2$ in the geometric flip graph connecting $T$ to a Delaunay triangulation of $(S,h)$. 
\end{theorem}

\begin{proof}
 Let $N_{v_1,v_2}$ be the number of paths from $v_1$ to $v_2$ shorter than $l=2\cdot\diam(T)$. From the previous lemma, we obtain that the $9g-9$ coordinates $m_i,t_i$, and $s_i$ of any such path $f$ are smaller than $c_h\cdot2\diam(T)$. It appears that, $\forall i, m_1=t_i+s_i$, $t_i=m_i+s_i$ or $s_i=m_i+t_i$~\cite[Lemma 6, p.64]{flp-tts-79}. So, if we fix $m_i$ and $t_i$ there are at most 3 possible $s_i$. Lemma~\ref{lem:paths} and~\ref{lem:pathlength} proves that there are $9(c_h\cdot 2\diam(T))^2$ potential paths for each coordinate set. We obtain a bound for $N_{v_1,v_2}$: $N_{v_1,v_2}\leq 9(c_h\cdot 2\diam(T))^2\cdot 3(c_h\cdot 2\diam(T))^{6g-6}$ and thus, there is a constant $C'_h$ such that $N_{v_1,v_2}\le C'_h\cdot\diam(T)^{6g-4}$. Since there are $1/2\cdot n^2$ possible sets $\{v_1,v_2\}$, we obtain the bound on the number of edges. 
\end{proof}

%\bibliographystyle{math}
%\bibliography{biblio}

\begin{thebibliography}{CGH{\etalchar{+}}}

\bibitem[Ber]{berger}
Marcel Berger.
\newblock {\em Geometry}.
\newblock Springer, 1996.

\bibitem[BS]{bs-1985}
Joan~S Birman and Caroline Series.
\newblock {Geodesics with bounded intersection number on surfaces are sparsely
  distributed}.
\newblock {\em Topology} {\bf 24}(1985), 217--225.

\bibitem[BDT]{bdt-hdcvd-14}
Mikhail Bogdanov, Olivier Devillers, and Monique Teillaud.
\newblock {Hyperbolic {D}elaunay complexes and {V}oronoi diagrams made
  practical}.
\newblock {\em Journal of Computational Geometry} {\bf 5}(2014), 56--85.

\bibitem[BIT]{cgal-bit-ht2}
Mikhail Bogdanov, Iordan Iordanov, and Monique Teillaud.
\newblock {\em {2D} Hyperbolic {Delaunay} Triangulations}.
\newblock {CGAL Editorial Board}, {4.14} edition, 2019.

\bibitem[BTV]{btv-dtosl-16}
Mikhail Bogdanov, Monique Teillaud, and Gert Vegter.
\newblock {Delaunay triangulations on orientable surfaces of low genus}.
\newblock In {\em Proceedings of the Thirty-second International Symposium on
  Computational Geometry}, pages 20:1--20:17, 2016.

\bibitem[BDG]{jdandco}
JD. Boissonnat, R.~Dyer, and A.~Ghosh.
\newblock {Delaunay Triangulation of Manifolds}.
\newblock {\em Foundations of Computational Mathematics} {\bf 18}(2018),
  399--431.

\bibitem[CT]{ct-dtced-16}
Manuel Caroli and Monique Teillaud.
\newblock {{Delaunay triangulations of closed {Euclidean} d-orbifolds}}.
\newblock {\em {Discrete \& Computational Geometry}} {\bf 55}(2016), 827--853.

\bibitem[CB]{casson-bleiler}
Andrew~J. Casson and Steven~A. Bleiler.
\newblock {\em Automorphisms of surfaces after {N}ielsen and {T}hurston},
  volume~9 of {\em London Mathematical Society Student Texts}.
\newblock Cambridge University Press, Cambridge, 1988.

\bibitem[CGH{\etalchar{+}}]{cghmsv-2010}
Carmen Cort{\'e}s, Clara~I Grima, Ferran Hurtado, Alberto M{\'a}rquez,
  Francisco Santos, and Jesus Valenzuela.
\newblock {Transforming triangulations on nonplanar surfaces}.
\newblock {\em SIAM Journal on Discrete Mathematics} {\bf 24}(2010), 821--840.

\bibitem[ES]{es-vda-86}
H.~Edelsbrunner and R.~Seidel.
\newblock {Voronoi diagrams and arrangements}.
\newblock {\em Discrete \& Computational Geometry} {\bf 1}(1986), 25--44.

\bibitem[FLP]{flp-tws-12}
Albert Fathi, Fran{\c{c}}ois Laudenbach, and Valentin Po{\'e}naru.
\newblock {\em Thurston's Work on Surfaces (MN-48)}.
\newblock Princeton University Press, 2012.

\bibitem[FLP{\etalchar{+}}]{flp-tts-79}
Albert Fathi, Fran{\c{c}}ois Laudenbach, Valentin Po{\'e}naru, et~al.
\newblock {\em Travaux de {T}hurston sur les surfaces}, volume 66--67 of {\em
  {A}st{\'e}risque}.
\newblock Soci{\'e}t{\'e} Math{\'e}matique de France, Paris, 1979.

\bibitem[Gar]{gardner}
Martin Gardner.
\newblock {\em Non-{E}uclidean Geometry, Chapter 4 of {The Colossal Book of
  Mathematics}}.
\newblock W. W. Norton \& Company, 2001.

\bibitem[HNU]{ferran}
F.~Hurtado, M.~Noy, and J.~Urrutia.
\newblock {Flipping edges in triangulations}.
\newblock {\em Discrete \& Computational Geometry} {\bf 3}(1999), 333--346.

\bibitem[IT]{it-idtbs-17}
Iordan Iordanov and Monique Teillaud.
\newblock {Implementing {Delaunay} triangulations of the {Bolza} surface}.
\newblock In {\em Proceedings of the Thirty-third International Symposium on
  Computational Geometry}, pages 44:1--44:15, 2017.

\bibitem[Mar]{m-aetim-69}
Gregorii~A Margulis.
\newblock {Applications of ergodic theory to the investigation of manifolds of
  negative curvature}.
\newblock {\em Functional analysis and its applications} {\bf 3}(1969),
  335--336.

\bibitem[Mas]{massey}
William~S. Massey.
\newblock {\em A basic course in algebraic topology}, volume 127 of {\em
  Graduate Texts in Mathematics}.
\newblock Springer-Verlag, New York, 1991.

\bibitem[Mir]{m-gnscg-08}
Maryam Mirzakhani.
\newblock {Growth of the number of simple closed geodesics on hyperbolic
  surfaces}.
\newblock {\em Annals of Mathematics} {\bf 168}(2008), 97--125.

\bibitem[Tah]{t-gtf-19}
Guillaume Tahar.
\newblock {Geometric triangulations and flips}.
\newblock {\em C. R. Acad. Sci. Paris, Ser. I} {\bf 357}(2019), 620–623.

\bibitem[Tra]{trainin}
J.~Trainin.
\newblock {An elementary proof of {P}ick's theorem}.
\newblock {\em Mathematical Gazette} {\bf 91}(2007), 536--540.

\end{thebibliography}

\newcommand{\etalchar}[1]{$^{#1}$}

\end{document}